\documentclass[a4paper,12pt]{article}
\usepackage[utf8]{inputenc}
\usepackage{amsmath, amssymb, amsthm, amscd}
\usepackage{mathrsfs,mathtools}
\usepackage[text={5.8in,8.8in},centering]{geometry}
\usepackage{color}
\usepackage{url}

\theoremstyle{plain}
\newtheorem{theorem}{Theorem}
\newtheorem{proposition}[theorem]{Proposition}
\newtheorem{lemma}[theorem]{Lemma}
\newtheorem{corollary}[theorem]{Corollary}

\theoremstyle{definition}
\newtheorem{definition}{Definition}

\theoremstyle{remark}
\newtheorem{remark}{Remark}[section]

\newcommand\bel{\begin{lemma}}
\newcommand\eel{\end{lemma}}
\newcommand\bep{\begin{proposition}}
\newcommand\eep{\end{proposition}}
\newcommand\bet{\begin{theorem}}
\newcommand\eet{\end{theorem}}
\newcommand\bex{\begin{example}}
\newcommand\eex{\end{example}}
\newcommand\bed{\begin{definition}}
\newcommand\eed{\end{definition}}
\newcommand\bea{\begin{assumption}}
\newcommand\eea{\end{assumption}}
\newcommand{\beq}{\begin{equation}}
\newcommand{\eeq}{\end{equation}}
\numberwithin{equation}{section}

\newcommand\lotimes{\mathop{\otimes}\limits}

\newcommand\loplus{\mathop{\oplus}\limits}

\def\bbbone{{\mathchoice {\rm 1\mskip-4mu l} {\rm 1\mskip-4mu l}
{\rm 1\mskip-4.5mu l} {\rm 1\mskip-5mu l}}}
\def\one{\bbbone}
\newcommand{\+}{{{+}{+}}}
\renewcommand{\t}{{\#}}
\renewcommand{\i}{\mathrm{i}}
\renewcommand{\d}{\mathrm{d}}
\newcommand{\e}{\mathrm{e}}

\newcommand{\h}{\mathrm{h}}
\newcommand{\p}{\mathrm{p}}
\newcommand{\reg}{\mathrm{reg}}

\newcommand{\cA}{\mathcal{A}}
\newcommand{\cS}{\mathcal{S}}
\newcommand{\pder}{\partial}
\newcommand{\qnd}{\mathrm{qnd}}

\renewcommand{\Re}{\mathrm{Re}\,}

\newcommand{\Tr}{\mathrm{Tr}\,}

\newcommand{\bbC}{\mathbb{C}}
\newcommand{\bbR}{\mathbb{R}}

\renewcommand{\bar}{\overline}
\newcommand{\Op}{\mathrm{Op}}
\newcommand{\Osc}{\mathrm{Osc}}
\newcommand{\Sym}{\mathrm{Sym}}
\newcommand{\nor}{\mathrm{nor}}

\newcounter{smallarabics}
\newenvironment{arabicenumerate}
{\begin{list}{{\normalfont\textrm{(\arabic{smallarabics})}}}
  {\usecounter{smallarabics}\setlength{\itemindent}{0cm}
   \setlength{\leftmargin}{5ex}\setlength{\labelwidth}{4ex}
   \setlength{\topsep}{0.75\parsep}\setlength{\partopsep}{0ex}
   \setlength{\itemsep}{0ex}}}
{\end{list}}

\newcounter{smallroman}

\newcommand{\ben}{\begin{arabicenumerate}}  
\newcommand{\een}{\end{arabicenumerate}}

\begin{document}

\title{Quantization of Gaussians}
\author{
  Jan Derezi\'{n}ski,\footnote{The financial support of the National Science
Center, Poland, under the grant UMO-2014/15/B/ST1/00126, is gratefully
acknowledged.}
  \hskip 3ex
  Maciej Karczmarczyk\footnotemark[\value{footnote}]
\\
Department of Mathematical Methods in Physics, Faculty of Physics\\
University of Warsaw,  Pasteura 5, 02-093, Warszawa, Poland\\
email: jan.derezinski@fuw.edu.pl\\
email: maciej.karczmarczyk@fuw.edu.pl}

\maketitle

\abstract{Our paper is devoted to the oscillator semigroup, which can be defined as the set of operators whose kernels are centered Gaussian. Equivalently, they can be defined as the the Weyl quantization of centered Gaussians. We use the Weyl symbol as the main parametrization of this semigroup. We derive  formulas for the tracial and operator norm of the Weyl quantization of Gaussians. We identify the subset of Gaussians, which we call {\em quantum degenerate}, where these  norms have a singularity.}

\vspace{0.5cm}
\textit{ Dedicated to the memory of Boris Pavlov.}

\section{Introduction}

Throughout our paper we will use the Weyl quantization, which is the most natural correspondence between quantum and classical states. For a function $a=a(x,p)$,
with  $x,p\in\bbR^d$, we will denote by $\Op(a)$ its Weyl quantization. Then function 
$a$ is called the Weyl symbol (or the Wigner function) of the operator $\Op(a)$.

The Heisenberg uncertainty relation says that one cannot  compress a state both in position and momentum without any limits. This is different than in classical mechanics, where in principle a state can have no dispersion both in position and momentum.

 One can ask what happens to a quantum state when we compress its Weyl symbol.
To be more precise, consider  the Gaussian function $\e^{-\lambda (x^2+p^2)}$, where $\lambda >0$ is an arbitrary parameter that controls the ``compression''. It is easy to compute the Weyl quantization of  $\e^{-\lambda (x^2+p^2)}$ and express it in terms  of the quantum harmonic oscillator
\beq H=\hat x^2+\hat p^2=\sum_{j=1}^d(\hat x_j^2+\hat p_j^2).\eeq
There are  3 distinct regimes of the parameter $\lambda $:
\beq
\Op\big(\e^{-\lambda (x^2+p^2)}\big)
=\begin{cases}(1-\lambda ^2)^{-d/2}\exp\Big(-\frac12\log\frac{(1+\lambda )}{(1-\lambda )}H\Big),&0<\lambda <1,\\
2^{-d}\one_{\{d\}}(H),&\lambda =1,\\
(\lambda ^2-1)^{-d/2}(-1)^{(H-d)/2}\exp\Big(-\frac12\log\frac{(1+\lambda )}{(\lambda -1)}H\Big),&
1<\lambda .
\end{cases}
\label{propo}\eeq
Thus, for $0<\lambda <1$, the quantization of the  Gaussian is proportional to  a thermal state of $H$. As $\lambda $ increases to $1$, it becomes ``less mixed''--its ``temperature'' decreases. At $\lambda =1$ it becomes  pure---its ``temperature'' becomes zero and it is the ground state of $H$. For $1<\lambda <\infty$, when we compress the Gaussian, it
 is no longer positive---due to the factor
 $(-1)^{(H-d)/2}$ it has eigenvalues with alternating signs.
Besides, it
becomes ``more and more mixed'', contrary to the naive classical picture.

Thus, at $\lambda =1$ we observe a kind of a ``phase transition'': For $0\leq\lambda <1$ the quantization of a Gaussian behaves more or less according to the classical intuition. For $1<\lambda$ the classical intuition stops to work---compressing the  classical symbol makes its quantization more ``diffuse''.

It is easy to compute the trace of (\ref{propo}):
\beq\Tr\Op\big(\e^{-\lambda (x^2+p^2)}\big)=\frac{1}{2^d\lambda ^{d}}.\label{propo1}\eeq
Evidently, (\ref{propo1}) does not see the ``phase transition'' at $\lambda =1$.
However, if we consider the trace norm, this phase transition appears---the trace norm of   (\ref{propo}) is differentiable except at  $\lambda =1$:
\beq\Tr\left|\Op\big(\e^{-\lambda (x^2+p^2)}\big)\right|=\begin{cases}
\frac{1}{2^d\lambda ^{d}}&\lambda \leq1,\\
\frac{1}{2^d},&1\leq \lambda 
.\end{cases}\label{propo2}\eeq
Note that (\ref{propo2}) can be viewed as a kind of quantitative ``uncertainty principle''.

Our paper is devoted to operators that can be written as the Weyl quantization of a (centered) Gaussian, more precisely, operators  of the form $a\Op(\e^{-A})$, where $A$ is a quadratic form with a strictly positive real part and $a\in\bbC$. Such operators form a semigroup called the \emph{oscillator semigroup}. We denote it by $\Osc_\+(\bbC^{2d})$. We also considered its subsemigroup, called the
\emph{normalized oscillator semigroup} and denoted $\Osc_\+^\nor(\bbC^{2d})$, which consists of operators
$\pm\sqrt{\det(\one+A\theta)}\Op(\e^{-A})$, where $\theta$ is $-\i$ times the symplectic form $\omega$.

The oscillator semigroup are closely related to the complex symplectic group $Sp(\bbC^{2d})$.  In particular, there exists a natural 2-1 epimorphism from 
$\Osc_\+^\nor(\bbC^{2d})$ onto $Sp_\+(\bbC^{2d})$, which is a certain natural subsemigroup of $Sp(\bbC^{2d})$.

The oscillator semigroup is closely related to the better known \emph{metaplectic group}, denoted $Mp(\bbR^{2d})$. The metaplectic group is generated by operators of the form $\pm\sqrt{\det(\one+ B\omega)}\Op(\e^{\i B})$, where $B$ is a real symmetric matrix. There exists a natural 2-1 epimorphism from $Mp(\bbR^{2d})$ to the real symplectic group $Sp(\bbR^{2d})$. 
Not all elements of the metaplectic group can be written as Weyl quantizations of a Gaussian.

The situation with the oscillator semigroup is somewhat different than with the metaplectic group. All elements of the oscillator semigroup are quantizations of a Gaussian, however, not all of them correspond to a (complex) symplectic transformation. Those, that do not, correspond to  quadratic forms $A$  satisfying $\det(\one+A\theta)=0$. 
We call such quadratic forms ``quantum degenerate''. Classically, they are of course nondegenerate.  Only their quantization is degenerate.
In particular, for a quantum degenerate $A$, the operator $\Op(\e^{-A})$  is not proportional to an element of  $\Osc_\+^\nor(\bbC^{2d})$. The set of quantum degenerate matrices can be viewed as a place where  some kind of a phase transition takes place in the oscillator semigroup.
For instance, as we show in our paper, the trace norm of $\Op(\e^{-A})$ depends smoothly on quantum nondegenerate $A$'s, however its smoothness typically breaks down 
at quantum degenerate $A$'s.

It is also natural to mention another type of an oscillator semigroup, which we denote $\Osc_+(\bbC^{2d})$. It is the semigroup generated by the operators of the form $a\Op(\e^{-A})$, where $A\geq0$. $\Osc_+(\bbC^{2d})$ contains both $\Osc_\+(\bbC^{2d})$ and $Mp(\bbR^{2d})$. It is in some sense the closure of $\Osc_\+(\bbC^{2d})$. We mention this semigroup only in passing, concentrating on $\Osc_\+(\bbC^{2d})$, which is easier, because, as we mentioned above, all elements of
 $\Osc_\+(\bbC^{2d})$ have Gaussian symbols. Note that the convenient notation $\+$ for $>0$ and $+$ for $\geq0$, which we use,  is borrowed from Howe \cite{Ho}.

Most of the time our discussion of the oscillator semigroup is
representation independent (without invoking a concrete Hilbert space on which $\Op(\e^{-A})$ acts).  
Perhaps the most obvious representation is the so-called Schr\"odinger representation, where the Hilbert space is $L^2(\bbR^d)$,  $\hat x$
is identified with the operator of multiplication by $x$ and $\hat p$ is $\frac{1}{\i}\partial_x$. Another possible representation is the Fock representation (or, which is essentially equivalent, the Bargmann-Fock representation, see e.g. \cite{DG}).
In both Schr\"odinger and Bargmann-Fock representations
the oscillator semigroup consists of operators with centered Gaussian kernels.

Let us now discuss the literature on operators with Gaussian kernels, or equivalently, on quantizations of Gaussians. Probably, the best known reference on this subject is a paper \cite{Ho} by Howe. In fact, we follow to some extent the  terminology from \cite{Ho}. His paper contains, for instance, a formula of composition of operators with Gaussian kernels, a criterion for positivity of such operators and the proof that there exists a 2-1 epimorphism from the normalized oscillator semigroup to a subsemigroup of $Sp(\bbC^{2d})$. 
Howe works mostly in the Schr\"odinger representation. Instead of the Weyl symbol, he  occasionally considers the so-called {\em Weyl transform}, which is essentially the Fourier transform of the Weyl symbol.

Another important work on the subject is a paper \cite{Hi} by Hilgert, who realised that the oscillator semigroup is isomorphic to a semigroup described by Bargmann, Brunet and Kramer (see \cite{Ba}, \cite{Br} and \cite{BrK}).
Hilgert uses mostly the {\em Fock-Bargmann  representation}.

The book  of Folland \cite{F}
contains a chapter on the oscillator semigroup, which sums up  the main points of \cite{Hi} and \cite{Ho}.
  
The existence of the ``phase transition'' at quantum degenerate positive Gaussians has been known for quite a long time, where the earliest reference we could find  is the paper \cite{U} by Unterberger.

Our paper  differs from \cite{Ho,Hi,F}
by using the Weyl quantization as the basic tool for the description of the oscillator semigroup.
It  is in some sense parallel to the presentation of the metaplectic group contained in Sect. 10.3 of \cite{DG}.
The Weyl quantization is in our opinion a natural tool in this context. First of all, it is symplectically invariant (unlike the Fock-Bargmann transform or the Schr\"odinger representation). Because of that, the analysis based on the Weyl quantization is particularly convenient and yields simple formulas. Secondly, the Weyl quantization allows us to make a direct contact with the quantum--classical correspondence principle.
  (This semiclassical aspect is hidden when one uses the Weyl transform, which is also symplectically invariant).

An operation, that we introduce, which we find interesting is the product $\t$ in the set of symmetric matrices. More precisely, it is defined so that
$\Op(\e^{-A})\Op(\e^{-B})$ is proportional to $\Op(-\e^{A\t B})$. Whenever defined, $\t$ is associative, however it is not always well defined.
$\t$ can be viewed as a semiclassical noncommutative distortion of the 
usual sum of square matrices. As we show,  quantum nondegenerate matrices with a positive part form a semigroup, which is essentially isomorphic to the oscillator semigroup.

Among new results obtained in our paper is a formula for the absolute value of and operator $\Op(\e^{-A})$, its trace norm and its operator norm.

There exists a close relationship between the set of complex matrices
equipped with $\#$ and the complex symplectic group. This relationship is quite intricate--it is almost a bijection, after removing some exceptional elements in both sets. One of new results of our paper is a detailed description of this relationship, see in particular Thm \ref{bijections-th}.

An interesting recent paper of Viola \cite{V} gives a formula for the norm of an element of the oscillator semigroup. 
Our formula for  $\big\|\Op(\e^{-A})\big\|$ is in our opinion  simpler than Viola's.

As an application of the formula for the trace norm of  $\Op(\e^{-A})$ we give a proof of the boundedness of the Weyl quantization with an explicit estimate of the of the operator norm. This result, which is a version of the so-called Calderon-Vaillancourt Theorem for the Weyl quantization, follows the ideas of Cordes \cite{C} and Kato \cite{K}, however the estimate of the norm seems to be new.

Elements of the oscillator semigroup can be viewed as exponentials of quantum quadratic Hamiltonians, that is $\e^{-\Op(H)}$, where $H$ is a classical quadratic Hamiltonian with a positive real part.  One example of such a Hamiltonian is
$\hat H_\psi:=\e^{\i\psi}\hat p^2+\e^{-\i\psi}\hat x^2$ for $|\psi|<\frac{\pi}{2}$, which is often called the {\em Davies harmonic oscillator}. It has been noted by a number of authors that this  operator has interesting, often counterintuitive  properties.
In particular, \cite{AV} and \cite{V} point out that $\e^{-z\hat H_\psi}$ can be defined as a bounded operator only for $z$ that  belong to a subset of the complex plane of a rather curious shape. We reproduce this result using methods developed in this article.

The  oscillator semigroup provides a natural framework for
a discussion of holomorphic semigroups $z\mapsto\e^{-z\Op(H)}$ associated with accretive quadratic Hamiltonians $\Op(H)$. We briefly discuss this issue at the end of our paper.

        \section{Notation}

        Let $L(\bbC^n)$ denote the set of $n\times n$ matrices.
        For $R\in L(\bbC^n)$ we will write $\bar R$, $R^\t$, resp. $R^*$
        for its complex conjugate,  transpose, resp.  Hermitian adjoint.
Elements of $\bbC^n$ are represented by column matrices, so that for $v,w\in\bbC^n$ the (sesquilinear) scalar product of $v$ and $w$ can be denoted by $v^*w$.

        $\sigma(R)$ will denote the spectrum of $R$.

We set
        \beq L^\reg(\bbC^n):=\{R\in L(\bbC^n)\ |\ R+\one \hbox{ is invertible }\}.\eeq
        For $R\in  L^\reg(\bbC^n)$, its  \emph{Cayley transform} is defined by
        \[c(R):=(\one-R)(\one+R)^{-1}.\]
        The Cayley transform is a bijection on 
$L^\reg(\bbC^n)$ and it is involutive, i.e. \beq c(c(R))=R.\eeq

For $A\in L(\bbC^n)$, we write $A>0$, resp. $A\geq0$ if
        \begin{align}
          v^* Av&>0,\quad v\in\bbC^n, \quad v\neq0,\\
          \text{resp.}\quad
          v^* Av&\geq0,\quad v\in\bbC^n.
        \end{align}

        $\Sym(\bbR^n)$, resp.        $\Sym(\bbC^n)$ denotes the set of symmetric real, resp. complex $n\times n$ matrices. We also set
        \begin{align} \Sym_+(\bbR^n)&:=\{A\in \Sym(\bbR^n)\mid A\geq0\},\\
          \Sym_\+(\bbR^n)&:=\{A\in \Sym(\bbR^n)\mid A>0\},\\
           \Sym_+(\bbC^n)&:=\{A\in \Sym(\bbC^n)\mid \Re A\geq0\},\\
          \Sym_\+(\bbC^n)&:=\{A\in \Sym(\bbC^n)\mid \Re A>0\}.
\end{align}
     		Note that $\Sym_\+(\bbC^n)$ is sometimes called the (generalised) \emph{Siegel upper half-plane}. It is sometimes denoted $S_n$ or $\mathfrak{S}_n$ \cite{Ho}. 

                The following proposition can be found in \cite{Ho}:
                
        \bep If  $A\in\Sym_\+(\bbC^n)$, then $A^{-1}$ exists
        and belongs to $\Sym_\+(\bbC^n)$.\eep

        \proof Let $A=A_{\rm r}+\i A_{\rm i}$ with
$A_{\rm r}\in \Sym_\+(\bbR^n)$, $A_{\rm i}\in \Sym(\bbR^n)$.
Let $B:=\sqrt{A_{\rm r}}$, $C:=B^{-1}A_{\rm i}B^{-1}$. Then $A=B(\one+\i C)B$ and
        $
        A^{-1}=B^{-1}(\one+\i C)^{-1}B^{-1}.$
        Clearly, $(\one+\i C)^{-1}\in \Sym_\+(\bbC^n)$. Hence
        $A^{-1}\in \Sym_\+(\bbC^n)$. \qed
        

        Every $n\times n$ symmetric matrix $A$ defines
    a quadratic form  on $\bbR^n$  by
        \beq
\bbR^n\ni y\mapsto y^\t Ay\in\bbC.\label{funio}\eeq
        We will often write $A$ for the function (\ref{funio}). Thus, in particular,
        \[\e^{-A}(y)=\e^{-y^\t Ay}.\]

We will often need to use the square root of a complex number $a$. If it is clear from the context that $a$ is positive and real, then
$\sqrt{a}$ will always denote the positive square root. If $a$ is a priori arbitrary, then $\pm\sqrt{a}$ will denote both values of the square root. If a given formula involves only one of possible values of the square root, then we will write \emph{$\epsilon\sqrt{a}$ where $\epsilon=1$ or $\epsilon=-1$.}

        \section{The Weyl quantization}
 Recall that 	for any $a\in\mathscr{S}'(\bbR^d\times\bbR^d)$
	\begin{equation}
		\Op(a)(x,\,y) = (2\pi)^{-d}\int a\Big(\frac{x+y}{2},\,p\Big)\e^{\i p(x-y)}\,\mathrm{d}p
	\end{equation}
is called the \emph{Weyl--Wigner quantization of the symbol $a$}, see e.g.  Section 18.5 of \cite{H3} or \cite{DG}. We can recover the symbol of a quantization from its distributional kernel by
	\begin{equation}\label{symbol-operatora-eq}
		a(x,\,p) = \int \Op(a)\Big(x + \frac{z}{2},\, x - \frac{z}{2}\Big)\,\e^{-\i zp}\,\mathrm{d} z.
		\end{equation}
		
For sufficiently nice functions $a,b$ we can define the \emph{star product} $*$ (sometimes called the \emph{Moyal star}) such that $\Op(a)\Op(b) = \Op(a*b)$ holds. On the level of symbols we have 
	\begin{equation}
		(a*b)(x,\,p) :=  \e^{\frac{\i}{2}(\pder_{x_1}\pder_{p_2}-\pder_{p_1}\pder_{x_2})}a(x_1,\, p_1)b(x_2,\, p_2) \Big|_{\begin{subarray}{l}x:=x_1=x_2 \\ p:=p_1=p_2    \end{subarray}}.
	\label{for}\end{equation}
	Write $y=\begin{bmatrix}x\\p\end{bmatrix}$, $\omega: = 
	\begin{bmatrix} 0 & \one_d \\ -\one_d & 0  
	\end{bmatrix},
        $ and $
	\theta: = \begin{bmatrix} 0 & -\i\one_d \\ \i\one_d & 0  
	\end{bmatrix}=-\i\omega.$
        One can rewrite (\ref{for}) in a more compact form:
        \begin{equation}
		(a*b)(y) =  \e^{-\frac{1}{2}\pder_{y_1}\theta\pder_{y_2}}a(y_1)b(y_2) \Big|_{y:=y_1=y_2}.
	\label{for1}\end{equation}
Here is an integral form of (\ref{for1}):
	\begin{equation}\label{moyal-integral}
	(a*b)(y) = \pi^{-2d}\int\mathrm{d}y_1\int\mathrm{d}y_2\,\e^{2(y-y_1)\theta(y-y_2)}a(y_1)b(y_2),
	\end{equation}
         (see e.g. \cite{DG}, Theorem 8.70.(4)).
        For the product of 3 symbols we have
        \begin{align}
          (a*b*c)(y)&=           \e^{-\frac{1}{2}\pder_{y_1}\theta\pder_{y_2}-\frac{1}{2}\pder_{y_1}\theta\pder_{y_3}-\frac{1}{2}\pder_{y_2}\theta\pder_{y_3}}a(y_1)b(y_2)c(y_3) \Big|_{y:=y_1=y_2=y_3}\\
          &\hspace{-9ex}=          \pi^{-3d}\int\mathrm{d}y_1\int\mathrm{d}y_2\int\mathrm{d}y_3\,
          \e^{(y-y_1)\theta(y-y_2)+(y-y_2)\theta(y-y_3)+(y-y_1)\theta(y-y_3)}
          \label{compi}
          \\
              &\hspace{-7ex}
              \times
\e^{-\frac12(y-y_1)\theta(y-y_1)-\frac12(y-y_2)\theta(y-y_2)-\frac12(y-y_3)\theta(y-y_3)} a(y_1)b(y_2)c(y_3)\Big|_{y=y_1=y_2=y_3}.\notag
          \end{align}

\section{Product $\t$}

Let $A,B\in\Sym(\bbC^{2d})$.
Suppose that
\beq
\text{the matrix }   \begin{bmatrix}\theta A\theta&-\theta\\
            \theta&\theta B\theta\end{bmatrix}\text{ is invertible.}\label{warunek}\eeq
We then  define $A\t B
\in\Sym(\bbC^{2d})$ by
\beq A\t B:=
\begin{bmatrix}-\one\\\one\end{bmatrix}^{\t}         \begin{bmatrix}\theta A\theta&-\theta\\
            \theta&\theta B\theta\end{bmatrix}^{-1} \begin{bmatrix}-\one\\\one\end{bmatrix}
.\label{oper}\eeq

            For the time being, the definition of the  product $\t$ may seem strange. As we will soon see in Section \ref{Oscillator semigroup}, it is motivated by the product of operators with Gaussian symbols.

            The following proposition gives a condition which guarantees that $A\t B$ is well defined.

            \bep (\ref{warunek}) holds
            iff the inverse of $(\one+A\theta B\theta)$ exists. We then have
            \begin{eqnarray}
        \begin{bmatrix}\theta A\theta&-\theta\\
            \theta&\theta B\theta\end{bmatrix}^{-1}
            &=&
            \begin{bmatrix}
              (\theta A\theta+ B^{-1})^{-1}&(\theta+\theta B\theta A\theta)^{-1}
              \notag\\
-(\theta+\theta A\theta B\theta)^{-1}&
              (\theta B\theta+ A^{-1})^{-1}
            \end{bmatrix}\\
            &=&\begin{bmatrix}
              B\theta(\one+A\theta B\theta)^{-1}\theta&
               (\one+B\theta A\theta)^{-1}\theta\\              
              -(\one+A\theta B\theta)^{-1}\theta&
             A\theta(\one+B\theta A\theta)^{-1}\theta
            \end{bmatrix},\label{invo1}\\
A\t B&=&             (\theta A\theta+ B^{-1})^{-1}+(\theta B\theta+ A^{-1})^{-1}\nonumber\\&&+
  (\theta+\theta A\theta B\theta)^{-1}
  -(\theta+\theta B\theta A\theta)^{-1}.\label{pro0}
            \end{eqnarray}
            \label{tag4}
            \eep

            \proof
It is well known how to compute an inverse of a $2\times2$
block matrix. This yields (\ref{invo1}), which implies  (\ref{pro0}).

            Clearly,
            \beq
            \theta(\one+A\theta B\theta)^\t\theta=(\one+B\theta A\theta).
            \eeq
            Therefore, the inverse of $(\one+A\theta B\theta)$ exists iff the inverse of $(\one+B\theta A\theta)$ exists. If this is the case, then all terms in   (\ref{invo1}) and (\ref{pro0}) are well defined. \qed

            \begin{proposition}
              The product $\t$ is associative. More precisely, if $A,B,C\in\Sym(\bbC^{2d})$ and $A\t B$, $B\t C$, $(A\t B)\t C$ and $A\t(B\t C)$ are well defined, then
              \beq (A\t B)\t C = A\t(B\t C).\label{form2}
              \eeq
              Besides,
              \begin{align}
                A\t0=0\t A=A,&&A\t(-A)=0,\\
                \bar{A\t B}=\bar{B}\t\bar{A},&& (-A)\t(-B)=-B\t A.
                \end{align}
\end{proposition}

            \proof We check that
            \begin{align}
            &(A\t B)\t C = A\t(B\t C)\\[2ex]
=& \begin{bmatrix}-\one\\0\\\one\end{bmatrix}^{\t}         \begin{bmatrix}\theta A\theta+\frac12\theta&-\frac12\theta&-\frac12\theta\\
    \frac12\theta&\theta B\theta+\frac12\theta&-\frac12\theta\\
\frac12\theta&\frac12\theta&\theta C\theta+\frac12\theta
  \end{bmatrix}^{-1} \begin{bmatrix}-\one\\0\\\one\end{bmatrix}
           .\end{align}
            (Compare with (\ref{compi})).
            This yields (\ref{form2}). The remaining statements are straightforward. \qed
            
            Note that it is useful to think of $\t$ as a noncommutative deformation of the addition. In fact, we have
            \beq A\t B= A+B+O(A^2+B^2).\eeq

\section{Quantum non-degenerate matrices}            

 Define
 \begin{eqnarray}
   \Sym^{\qnd}(\bbC^{2d})&:=&\{A\in
   \Sym(\bbC^{2d})\ :\ \det(\one+A\theta)\ne0\},\\
      \Sym_\+^{\qnd}(\bbC^{2d})&:=&\{A\in
      \Sym_\+(\bbC^{2d})\ :\ \det(\one+A\theta)\ne0\},\\
         \Sym^{\qnd}(\bbR^{2d})&:=&\{A\in
         \Sym(\bbR^{2d})\ :\ \det(\one+A\theta)\ne0\},\\
            \Sym_\+^{\qnd}(\bbR^{2d})&:=&\{A\in
 \Sym_\+(\bbR^{2d})\ :\ \det(\one+A\theta)\ne0\}.\end{eqnarray}
 ($\qnd$ stands for \emph{quantum non-degenerate}).

 There are several equivalent formulas for the product  (\ref{oper}). It is actually not so obvious to pass from one of them to another. In the following proposition  we give a few of them. 
            \bep Let $A,B\in \Sym^{\qnd}(\bbC^{2d})$ such that
            $(\one+A\theta B\theta)^{-1}$ exists. Then
\begin{eqnarray}
  A\t  B&=&
c\big(c(A\theta)c(B\theta)\big)\theta
\label{produ}\\
&=&
  (\one+A\theta)^{-1}(A\theta+B\theta)(\one+A\theta B\theta)^{-1}(\one+A\theta)\theta\label{pro1}\\
  &=&
  (\one+B\theta)(\one+A\theta B\theta)^{-1}(A\theta+B\theta)(\one+B\theta)^{-1}\theta\label{pro2}\\
    &=&
(\one-A\theta)(\one+B\theta A\theta)^{-1}(A\theta+B\theta)(\one-A\theta)^{-1}\theta\label{pro3a}\\
&=&
  (\one-B\theta)^{-1}(A\theta+B\theta)(\one+B\theta A\theta)^{-1}(\one-B\theta)\theta.\label{pro4}
\end{eqnarray}
We have
\beq
\one+A\theta B\theta=(\one+A\theta)(\one+A\t B\theta)^{-1}(\one +B\theta),
\label{prio}\eeq
and 
$A\t B\in \Sym^\qnd(\bbC^{2d})$.
\eep

\proof
To see (\ref{produ}), it is enough to show that
\beq c(A\t B\theta)=c(A\theta)c(B\theta).\eeq
 (\ref{pro0})  can be rewritten as
\begin{eqnarray*}
  A\t B&=&B\theta(\one+A\theta B\theta)^{-1}\theta+A\theta(\one+B\theta A\theta)^{-1}\theta\\
  &&+(\one+A\theta B\theta)^{-1}\theta-(\one+B\theta A\theta)^{-1}\theta\\
  &=&(\one+B\theta)(\one+A\theta B\theta)^{-1}\theta
  -(\one-A\theta)(\one+B\theta A\theta)^{-1}\theta.
\end{eqnarray*}
Therefore,
\begin{eqnarray}
  \one-A\t B\theta&=&(A\theta-\one)B\theta(\one+A\theta B\theta)^{-1}+
  (\one-A\theta)(\one+B\theta A\theta)^{-1}\notag\\
  &=&
  (\one-A\theta)(\one+B\theta A\theta)^{-1}(\one-B\theta);\label{tag1}
 \\
  \one+A\t B\theta&=&(\one+B\theta)A\theta(\one+B\theta A\theta)^{-1}+
  (\one+B\theta)(\one+A\theta B\theta)^{-1}\notag\\
  &=&
  (\one+B\theta)(\one+A\theta B\theta)^{-1}(\one+A\theta).\label{tag2}
  \end{eqnarray}
Hence,
\begin{eqnarray*}
&&  c(A\t B\theta)\\&=&(\one-A\theta)(\one+B\theta A\theta)^{-1}(\one-B\theta)(\one+A\theta)^{-1}
  (A\theta B\theta+\one)(\one+B\theta)^{-1}\\
  &=&(\one-A\theta)(\one+B\theta A\theta)^{-1}(\one-B\theta)\big(B\theta+(\one+A\theta)^{-1}(\one-B\theta)\big)
  (\one+B\theta)^{-1}\\
   &=&(\one-A\theta)(\one+B\theta A\theta)^{-1}\big(B\theta+(\one-B\theta)(\one+A\theta)^{-1}\big)
  (\one-B\theta)
  (\one+B\theta)^{-1}\\
     &=&(\one-A\theta)(\one+B\theta A\theta)^{-1}(\one+B\theta A\theta)(\one+A\theta)^{-1}
  (\one-B\theta)
  (\one+B\theta)^{-1}\\
  &=&c(A\theta)c(B\theta).
\end{eqnarray*}
Thus  (\ref{produ}) is proven.

Next note that
\begin{align}
  c(A\theta)c(B\theta)&=(\one +A\theta)^{-1}(\one-A\theta)(\one-B\theta)  (\one+B\theta)^{-1}\\
  &=(\one +A\theta)^{-1}(\one-A\theta-B\theta+A\theta B\theta)  (\one+B\theta)^{-1}.\end{align}
Therefore,
\begin{align}
\one-  c(A\theta)c(B\theta)&=2(\one +A\theta)^{-1}(A\theta+B\theta)  (\one+B\theta)^{-1}\label{inser1}\\
\one+  c(A\theta)c(B\theta)
&=2(\one +A\theta)^{-1}(\one+A\theta B\theta)(\one+B\theta)^{-1}.\label{inser2}\end{align}
Next we insert (\ref{inser1}) and (\ref{inser2}) into
\begin{eqnarray}
A\t B\ =\   c\big(c(A\theta)c(B\theta)\big)\theta&=&
\big(\one- c(A\theta)c(B\theta)\big)\big(\one+ c(A\theta)c(B\theta)\big)^{-1}
\theta
  \label{pro5}\\
  &=&\big(\one+c(A\theta)c(B\theta)\big)^{-1}\big(\one- c(A\theta)c(B\theta)\big)\theta.\label{pro6}
  \end{eqnarray}
obtaining (\ref{pro1}), resp. (\ref{pro2}).

We know that $A\t B$ is symmetric. Applying the transposition to (\ref{pro1}), resp. (\ref{pro2}) we obtain  (\ref{pro3a}), resp. (\ref{pro4}), where we use $\theta^\t=-\theta$, $A^\t=A$, $B^\t=B$.

(\ref{prio}) is proven in (\ref{tag2}). This implies that $\one + A\t B\theta$ is invertible. Hence $A\t B\in  \Sym^\qnd(\bbC^{2d})$.
\qed

$\Sym^{\qnd}(\bbC^{2d})$  equipped with
(\ref{oper}) is not a semigroup. It is enough to see that for $A=B=\begin{bmatrix}\i&0\\0&\i\end{bmatrix}$ we have $\one+A\theta B\theta=0$, so $A\t B$ is not defined.
  
            \begin{proposition}	
            $\Sym_\+(\bbC^{2d})$ is a semigroup.
\label{tag3}\end{proposition}

              	\begin{proof}	Let $A$ and $B$ belong to $\Sym_\+(\bbC^{2d})$.
 The matrix $          \begin{bmatrix}\theta A\theta&-\theta\\
   \theta&\theta B\theta\end{bmatrix}$ belongs to $\Sym_\+(\bbC^{2d})$. Hence, so does its inverse. Therefore, (\ref{oper}) also belongs
   to $\Sym_\+(\bbC^{2d})$. This shows that
 $A\t B$ is well defined and belongs to $ \Sym_\+(\bbC^{2d})$.
            \end{proof}

            \begin{proposition}
              $\Sym_\+^\qnd(\bbC^{2d})$ is also a semigroup.
            \end{proposition}
            
            \begin{proof}	Let $A$ and $B$ belong to $\Sym_\+^\qnd(\bbC^{2d})$. We already know that $A\t B$ is well defined, and hence $\one+A\theta B\theta$ is invertible (see  Prop. \ref{tag4}).
                Using (\ref{prio}) and the invertibility of $\one+A\theta$, $\one+B\theta$,  we can conclude that $\one+A\t B\theta$ is invertible. Hence $A\t B\in \Sym_\+^\qnd(\bbC^{2d})$.                
\end{proof}
              
\section{Oscillator semigroup}
\label{Oscillator semigroup}

       Following \cite{Ho,F}, the {\em oscillator semigroup} $\Osc_\+(\bbC^{2d})$ is defined as the set of operators on $L^2(\bbR^d)$
        whose Weyl symbols are centered Gaussian, that is operators of the form
        $a\Op(\e^{-A})$, where $a\in\bbC$, $A\in\Sym_\+(\bbC^{2d})$ and
        $A(x,p)=\begin{bmatrix}x\\p\end{bmatrix}^\t A
        \begin{bmatrix}x\\p\end{bmatrix}$.
(In \cite{Ho}, this semigroup is  denoted $\Omega$).

There are several equivalent characterizations of $\Osc_\+(\bbC^{2d})$. Here is one of them:

          \bep $\Osc_\+(\bbC^{2d})$ equals the set of operators on $L^2(\bbR^d)$ with centered Gaussian kernels. More precisely, the integral kernel of
          $a\Op(\e^{-A})$ for $A = \begin{bmatrix} B & D \\ D^\t & F\end{bmatrix}$ is
          \[c\e^{-C(x,y)},\]
          where
          \beq c=\frac{2^{-d}a}{\sqrt{\det(F)}},\eeq  
          \begingroup\makeatletter\def\f@size{9}\check@mathfonts
          
           \begin{eqnarray*}C(x,y)&=& -\frac{1}{4}\begin{bmatrix}x \\ y\end{bmatrix}^{\t}\begin{bmatrix}1 & 1 \\ 1 & -1\end{bmatrix}\begin{bmatrix}B-D F^{-1}D^\t &  - \i DF^{-1} \\-\i F^{-1}D^\t & F^{-1}\end{bmatrix}\begin{bmatrix}1 & 1 \\ 1 & -1\end{bmatrix}\begin{bmatrix}x \\ y\end{bmatrix}\\
          &\hspace{-18ex}=& \hspace{-10ex}-\frac{1}{4}\begin{bmatrix}x \\ y\end{bmatrix}^{\t}\begin{bmatrix}B-D F^{-1}D^\t-\i DF^{-1} -\i F^{-1} D^\t +F^{-1} &   B-D F^{-1}D^\t+\i DF^{-1} -\i F^{-1} D^\t -F^{-1} \\ B-D F^{-1}D^\t-\i DF^{-1} +\i F^{-1} D^\t -F^{-1} & B-D F^{-1}D^\t+\i DF^{-1} +\i F^{-1} D^\t +F^{-1}\end{bmatrix}\begin{bmatrix}x \\ y\end{bmatrix} .\end{eqnarray*}\endgroup
          \eep

          \medskip

\proof The formula follows by elementary Gaussian integration. The detailed computations can be found in \cite{F}. \qed

            \begin{proposition}		Let $A$ and $B$ belong to $\Sym_\+(\bbC^{2d})$.
              Then the following product formula holds:
		\begin{equation}\label{star-product-formula}
		  \Op(\e^{-A})\Op(\e^{-B}) = \frac{\epsilon}{\sqrt{\det(A\theta B\theta +\one)}}\Op(\e^{-A\t B}),
		\end{equation}
                where
                $\epsilon=1$ or $\epsilon=-1$.
                Consequently, $\Osc_\+(\bbC^{2d})$ is a semigroup
                and
                \beq
                \Osc_\+(\bbC^{2d})\ni c\Op(\e^{-A})\mapsto A\in\Sym_\+(\bbC^{2d})
                \eeq
                is an epimorphism.
	\end{proposition}

            \begin{proof}
              Formula \eqref{moyal-integral} assures us that
		\begin{eqnarray}\label{star-pomoc}
                  &&			(\e^{-y^\t A y}*\e^{-y^\t By})(y)\\&=& \pi^{-2d}\int\mathrm{d}y_1\int\mathrm{d}y_2 \exp\big(-2(y-y_2)\theta(y-y_1)-y_1^\t A y_1-y_2^\t By_2\big)\\&=&
                  \pi^{-2d}\int\mathrm{d}y_1\int\mathrm{d}y_2 \exp\left(-\begin{bmatrix}
       y_1\\y_2\end{bmatrix}^\t\begin{bmatrix}A&-\theta\\\theta&B\end{bmatrix}
                  \begin{bmatrix}
       y_1\\y_2\end{bmatrix}-2\begin{bmatrix}
       y_1\\y_2\end{bmatrix}^\t\begin{bmatrix}
       -\theta y\\\theta y\end{bmatrix}\right)
       \nonumber\\
       &=&
       \det\begin{bmatrix}A&-\theta\\\theta&B\end{bmatrix}^{-1/2}
       \exp\left(
      \begin{bmatrix}
        -\theta y\\\theta y\end{bmatrix}^\t \begin{bmatrix}A&
          -\theta\\\theta&B\end{bmatrix}^{-1}
         \begin{bmatrix}
       -\theta y\\\theta y\end{bmatrix}\right)
             .   \end{eqnarray}
		Then we check that
                \beq
                \det\begin{bmatrix}A&-\theta\\\theta&B\end{bmatrix}
                =\det(\one+A\theta B\theta),\eeq
                \beq
                     \begin{bmatrix}
 -\theta y\\\theta y\end{bmatrix}^\t \begin{bmatrix}A&-\theta\\\theta&B\end{bmatrix}^{-1}
    \begin{bmatrix}
           -\theta y\\\theta y\end{bmatrix}
           =
-\begin{bmatrix}
       -y\\ y\end{bmatrix}^\t     
           \begin{bmatrix}\theta A\theta&-\theta\\
\theta&\theta B\theta\end{bmatrix}^{-1}\begin{bmatrix}-y\\y\end{bmatrix}
 .\eeq
        \end{proof}

Again, following \cite{Ho,F}, we introduce
the \emph{normalized oscillator semigroup}, denoted $\Osc_\+^\nor(\bbC^{2d})$, as 
    \[\left\{\pm\sqrt{\det(\one+A\theta)}\Op(\e^{-A}) \, | \, A\in\Sym_\+^\qnd(\bbC^{2d})\right\}.\]
(In \cite{Ho}, this semigroup is  denoted $\Omega^0$).

\bep $\Osc_\+^\nor(\bbC^{2d})$ is a subsemigroup of $\Osc_\+(\bbC^{2d})$
and                \beq
\Osc_\+^\nor(\bbC^{2d})\ni
\pm\sqrt{\det(\one+A\theta)}\Op(\e^{-A})\mapsto A\in\Sym_\+^\qnd(\bbC^{2d})
               \label{epi2} \eeq
                is a 2-1 epimorphism of semigroups.
                \label{semig}
	        \eep

\proof
It is enough to check that
	\begin{eqnarray}
	  &&\sqrt{\det(\one + A\theta)}\Op(\e^{-A})\sqrt{\det(\one + B\theta)}\Op(\e^{-B})\\&=&
          \epsilon\sqrt{\det(\one + A\t B\theta)}\Op(\e^{-A\t B}),
	\end{eqnarray}
        where $\epsilon=1$ or $\epsilon=-1$.
        Indeed, (\ref{prio}) implies
        \beq
\det(\one+A\theta B\theta)=\det(\one+A\theta)\det(\one+A\t B\theta)^{-1}\det(\one +B\theta).\eeq
Now we need to use
(\ref{star-product-formula}).
\qed

\section{Positive elements of the oscillator semigroup}

         We  define
         \begin{align}
           \Sym_\p(\bbR^{2d})&:=\{A\in \Sym_\+(\bbR^{2d})\mid \sigma(A\theta)\subset[-1,1]\},\\
          \Sym_\p^\qnd(\bbR^{2d})&:=             \{A\in \Sym_\p(\bbR^{2d})\mid \det(A\theta+\one)\neq0\}.
               \end{align}

        \bep Let $a\in\bbC$ and $A\in\Sym_\+(\bbC^{2d})$. Then
        \ben
\item $\big(a\Op(\e^{-A})\big)^*=\bar a\Op(\e^{-\bar A})$.
      \item $a\Op(\e^{-A})$ is Hermitian iff $a\in\bbR$ and $A\in\Sym_\+(\bbR^{2d})$.
        \item 
          $a\Op(\e^{-A})$ is positive iff $a>0$,
          $A\in\Sym_\p(\bbR^{2d})$.
          \een\label{prio1}\eep

          \proof (1) and (2) are follow by the obvious identity $\Op(a)^*=\Op(\bar a)$.

           Let us prove (3).
$A$ is a positive definite real matrix and $\omega$ is a symplectic matrix. It is well known, that they can be simultaneously diagonalized, that is, one can find a basis of $\bbR^{2d}$ such that if we write 
          $\bbR^{2d}=\loplus_{i=1}^d\bbR^2$, then $\omega$ is the direct sum of $\begin{bmatrix}0&1\\-1&0\end{bmatrix}$ and
            $A$ is the direct sum of
            $\begin{bmatrix}\lambda_i&0\\0&\lambda_i\end{bmatrix}$
              with $\lambda_i>0$. After an appropriate metaplectic transformation, we can represent the Hilbert space $L^2(\bbR^d)$ as
              $\lotimes_{i=1}^d L^2(\bbR)$ and $\Op(\e^{-A})$ can be represented as $\lotimes_{i=1}^d\Op\big(\e^{-\lambda_i(x_i^2+p_i^2)}\big)$.
              Next we use (\ref{propo}) to see that the positivity of $\Op(\e^{-A})$ is equivalent to $\lambda_i\leq1$, $i=1,\dots,d$, which in turn is equivalent to $\sigma(A\theta)\subset [-1,1]$ (the eigenvalues of $A\theta$ are of the form $\pm \lambda_i$). \qed

              \bep $\Sym_\p^\qnd(\bbR^{2d})=
\{A\in \Sym_\+(\bbR^{2d})\mid \sigma(A\theta)\subset]-1,1[\}.$
            \label{symp-in-sym+}  \eep

              \proof We use the basis that is mentioned at the end of the proof of Prop. \ref{prio1}. \qed

         \bep $\det(\one+A\theta)=\bar{\det(\one+\bar A\theta)}$. Consequently, 
         $\Sym^{\qnd}(\bbC^{2d})$ and $\Sym_\+^{\qnd}(\bbC^{2d})$ are invariant with respect to complex conjugation. \eep

         \proof We use \begin{eqnarray}
           \theta(\one+A\theta)\theta&=&\one+\theta A,\\
           (\one+\theta A)^\t&=&\one-A\theta,\\
           \bar{\one-A\theta}&=&\one+\bar A\theta.\end{eqnarray}
         \qed
         


        \bet\label{AtA}
\ben\item If $A\in \Sym_\+(\bbC^{2d})$, then
$\bar A\t A\in\Sym_\p(\bbR^{2d})$.
\item  If $A\in \Sym_\+^\qnd(\bbC^{2d})$, then
  $\bar A\t A\in\Sym_\p^\qnd(\bbR^{2d})$. \een
\eet

\proof (1): Let $A\in\Sym_\+(\bbC^{2d})$.
Then
 \begin{eqnarray}
              \Op(\e^{-A})^*\Op(\e^{-A})&=&
              \frac{1}{\sqrt{\det(\one+\bar A\theta A\theta)}}\e^{-\bar A\t A}
                       \end{eqnarray}
     is a positive operator. Therefore,
     by  Proposition \ref{prio1} (3),     $\bar A\t A\in        \Sym_\p(\bbR^{2d})$.
     
           (2):  $\Sym_\+^\qnd(\bbC^{2d})$ is a semigroup invariant wrt the conjugation, and hence $\bar A\t A\in
     \Sym_\+^\qnd(\bbC^{2d})$.
     By (1), $\bar A\t A\in \Sym_\p(\bbR^{2d})$.
     But by definition
           $\Sym_\p^\qnd(\bbR^{2d})
           =\Sym_\p(\bbR^{2d})\cap \Sym_\+^\qnd(\bbC^{2d})$. 
           \qed

         \section{Complex symplectic group}

A linear operator $R$ on $\bbR^{2d}$ is called \emph{symplectic} if
          \beq R^\t \omega R=\omega.\label{sympl}\eeq
          The set of
symplectic
operators  on $\bbR^{2d}$ will be denoted $Sp(\bbR^{2d})$. It is the well known \emph{symplectic group} in dimension $2d$.

In our paper a more important role is played by the complex version of the symplectic group. More precisely, we will say that a complex linear operator $R$ on $\bbC^{2d}$ is symplectic if (\ref{sympl}) holds.          (Of course, we can replace
          $\omega$ in (\ref{sympl}) with $\theta$).
The set of complex symplectic
operators  on $\bbC^{2d}$ will be denoted $Sp(\bbC^{2d})$. It is also a group, called the \emph{complex symplectic group} in dimension $2d$.

We define
\begin{align}
  Sp_+(\bbC^{2d})&:=\{R\in Sp(\bbC^{2d})\mid R^* \theta R\leq \theta\},\\
  Sp_\+(\bbC^{2d})&:=\{R\in Sp(\bbC^{2d})\mid R^* \theta R< \theta\}.
  \end{align}
  $Sp_+(\bbC^{2d})$ and $Sp_\+(\bbC^{2d})$
  are semigroups satisfying
  \begin{align} Sp(\bbR^{2d})\cap Sp_\+(\bbC^{2d})&=\emptyset,\\
    Sp_\+(\bbC^{2d})&\subset Sp_+(\bbC^{2d}),\\
    Sp(\bbR^{2d})&\subset Sp_+(\bbC^{2d}).
\end{align}

          We also set
          \begin{eqnarray}
Sp_\h(\bbC^{2d})&:=&\{R\in
Sp(\bbC^{2d}) :\ \bar R=R^{-1}\}\\
&=&\{R\in
            Sp(\bbC^{2d}) :\ R^*\theta=\theta R\},\\
                       Sp_{\mathrm{p}}(\bbC^{2d})&:=&\{R\in
 Sp_\h(\bbC^{2d}) :\ \sigma( R)\subset]0,\infty[\}.\end{eqnarray}

Below we state a few properties of $Sp_\+(\bbC^{2d})$ and  $Sp_\p(\bbC^{2d})$. It will be convenient to defer their proofs to the next section.

\bep $Sp_\p(\bbC^{2d})\subset 
Sp_\+(\bbC^{2d})$.
\label{pro1-}\eep

Let  $t>0$.
Note that $\bbC\backslash]-\infty,0]\ni z\mapsto z^t\in\bbC$ is a well defined holomorphic function. In the proposition below $\sigma(R)\subset]0,\infty[$, therefore $R^t$ is well defined.
          
          \bep Let $R\in Sp_{\rm p}(\bbC^{2d})$. Then $R^t\in Sp_{\rm p}(\bbC^{2d})$.
        \label{proi0}  \eep

\bep Let $R\in Sp_\+(\bbC^{2d})$. Then $\bar R^{-1}R\in Sp_\p(\bbC^{2d}).$
\label{proi2}\eep

The next result, which
is an analog of the polar decomposition, was noted by Howe (see \cite{Ho}, Proposition (23.7.2)):
\begin{proposition}\label{matrix-decomp}
	Every $R\in Sp_\+(\bbC^{2d})$ may be decomposed in the following way:
	\begin{equation}
		R = TS,
	\end{equation}
	where $T:=\bar{R}\sqrt{\bar{R}^{-1}R}\in Sp(\bbR^{2d})$ and $S:=\sqrt{\bar{R}^{-1}R}\in Sp_\p(\bbC^{2d})$. 
\eep

\section{Relationship between $\Sym$  and symplectic group}

Let us define
\begin{align}
  Sp^\reg(\bbC^{2d})&=\left\{R\in Sp(\bbC^{2d})\ |\ R+\one \hbox{ is invertible }\right\},\\
  Sp_\h^\reg(\bbC^{2d})&=\left\{R\in Sp_\h(\bbC^{2d})\ |\ R+\one \hbox{ is invertible }\right\}.
  \end{align}

\bet\label{bijections-th}
\ben\item
$
\Sym^\qnd(\bbC^{2d})\ni A\mapsto c(A\theta)\in Sp^\reg(\bbC^{2d})$ is a bijection.
Its inverse is
\beq
 Sp^\reg(\bbC^{2d})\ni R\mapsto c(R)\theta\in
\Sym^\qnd(\bbC^{2d}).\eeq
Besides, if $A,B\in \Sym^\qnd(\bbC^{2d})$ and $A\t B\in \Sym^\qnd(\bbC^{2d})$ is well defined, then
\beq c(A\t B\theta)=c(A\theta)c(B\theta).\label{produ1}\eeq
\item
$ \Sym_\+^\qnd(\bbC^{2d})\ni A\mapsto c(A\theta)\in Sp_\+(\bbC^{2d})$ is an isomorphism of semigroups.
\item
$\Sym^\qnd(\bbR^{2d})\ni A\mapsto c(A\theta)\in Sp_\h^\reg(\bbC^{2d})$
  is a bijection.
\item
$\Sym_{\mathrm{p}}^\qnd(\bbR^{2d})\ni A\mapsto c(A\theta)\in Sp_\p(\bbC^{2d})$
  is a bijection.
          \een\eet

          \proof (1): Let $A\in \Sym^\qnd(\bbC^{2d})$. Then,
          \begin{eqnarray*}
            c(A\theta)^\t \theta c(A\theta)
            &=&
            (\one-
            \theta A)^{-1}(\one+
            \theta A)\theta(\one-A\theta)(\one+A\theta)^{-1}
            \\&=&
            (\one-
            \theta A)^{-1}(\one-\theta  A\theta A\theta)
            (\one+A\theta)^{-1}\\&=&
            (\one-
            \theta A)^{-1}(\one-
            \theta A)\theta(\one+A\theta)(\one+A\theta)^{-1}\ =\ \theta.
            \end{eqnarray*}
          Hence, $ c(A\theta)\in Sp(\bbC^{2d})$.

          Conversely, let $R\in Sp^\reg(\bbC^{2d})$. Then
    \begin{eqnarray*}
      \Big((\one-R)(\one+R)^{-1}\theta\Big)^\t 
      &=&-\theta(\one+R^\t )^{-1}(\one-R^\t )
      \\
  =    -(\theta+R^\t \theta)^{-1}(\one-R^\t )&=&
      -\big(\theta(\one+R^{-1})\big)^{-1}(\one-R^\t )\\
      =-(\one+R^{-1})^{-1}(\theta-\theta R^\t )
      &=&-(\one+R^{-1})^{-1}(\one- R^{-1})\theta
      \\=(\one+R)^{-1}(\one-R)\theta.&&
      \end{eqnarray*}
Hence, $c(R)\theta\in\Sym(\bbC^{2d})$.
   
Clearly, $A\theta+\one$ is invertible  iff $c(A\theta) \in L^\reg(\bbC^{2d})$. Thus $
\Sym^\qnd(\bbC^{2d})\ni A\mapsto c(A\theta)\in Sp^\reg(\bbC^{2d})$ is a bijection.

To see (\ref{produ1}) it is enough to use  (\ref{produ}).

          (2): We have
\begin{eqnarray*}
            c(A\theta)^*\theta c(A\theta)
            &=&
            (\one+
            \theta \bar A)^{-1}(\one-
            \theta \bar A)\theta(\one-A\theta)(\one+A\theta)^{-1}
            \\&=&
            (\one+
            \theta \bar A)^{-1}(\one-\theta \bar A\theta-\theta A\theta+\theta \bar A\theta A\theta)
            (\one+A\theta)^{-1}\\&=&
          \theta-2  (\one+
            \theta \bar A)^{-1}\theta(\bar A+A)\theta (\one+A\theta)^{-1}.
            \end{eqnarray*}
Thus,
\beq         c(A\theta)^*\theta c(A\theta)<\theta\eeq
iff $\bar A+A>0$.
Hence $ \Sym_\+^\qnd(\bbC^{2d})\ni A\mapsto c(A\theta)\in Sp_\+(\bbC^{2d})$ is a bijection. It is a homomorphism because of  (\ref{produ1}).

(3): Let $A\in \Sym^\qnd(\bbR^{2d})$. Then
\beq
  \bar{c(A\theta)}=\frac{\one+A\theta}{\one-A\theta}=c(A\theta)^{-1}.
  \eeq
  Hence, $c(A\theta)\in Sp_\h(\bbC^{2d})$.

  Conversely, let  $R\in Sp_\h^\reg(\bbC^{2d})$.
  Then
      \begin{eqnarray}
      \bar{(\one-R)(\one+R)^{-1}\theta}&=&
      -(\one-\bar R)(\one+\bar R)^{-1}\theta\\
      = -(\one- R^{-1})(\one+ R^{-1})^{-1}\theta&=&(\one- R)(\one+ R)^{-1}\theta.
      \end{eqnarray}
      Hence, $c(R)\theta \in \Sym(\bbR^{2d})$.
      
(4): Clearly, \[\lambda\in]-1,1[\quad\text{ iff }\quad \frac{1-\lambda}{1+\lambda}\in]0,\infty[.\] Therefore, \[\sigma(A\theta)\subset]-1,1[\quad\text{ iff }\quad \sigma\big(c(A\theta)\big)\subset]0,\infty[.\]
Then we use the characterisation of $\Sym_\p^\qnd(\bbR^{2d})$ given in Prop. \ref{symp-in-sym+}.  \qed

\begin{proof}[Proof of Proposition \ref{pro1-}.]
	Let $R\in Sp_\p(\bbC^{2d})$. By Theorem \ref{bijections-th} (4), $c(R)\theta\in\Sym_\p^\qnd(\bbR^{2d})$. Proposition \ref{symp-in-sym+} implies that $c(R)\theta\in\Sym_\+^{\qnd}(\bbR^{2d})$. Now Theorem \ref{bijections-th}~(2) shows that $R= c\big((c(R)\theta)\theta\big)\in Sp_\+(\bbC^{2d})$.
\end{proof}

\begin{proof}[Proof of Proposition \ref{proi0}]
  Let $R\in Sp_\p(\bbC^{2d})$.

Functional calculus of operators is invariant wrt  similarity transformations. Therefore,  $R^{\t(-1)}=\theta R\theta^{-1}$ implies
  $R^{\t(-t)}=\theta R^t\theta^{-1}$. Hence $R^t\in  Sp(\bbC^{2d})$.
  
  $\bar R=R^{-1}$ implies $\bar R^t=(R^t)^{-1}$. Hence, $R^t\in Sp_\h(\bbC^{2d})$.

  $\sigma(R)\subset]0,\infty[$ implies  $\sigma(R^t)\subset]0,\infty[$.
          Hence $R^t\in Sp_\p(\bbC^{2d})$.
\end{proof}

\begin{proof}[Proof of Proposition \ref{proi2}.]
	Theorem \ref{bijections-th}~(2) assures us that we can find a matrix $A\in\Sym_\+^\qnd(\bbR^{2d})$, such that $c(A\theta)=R$.
	By Theorem \ref{AtA} (2), $\bar{A}\t A\in \Sym_\p^\qnd(\bbR^{2d})$.
        Now we may use Theorem \ref{bijections-th} (4) to see that $c(\bar{A}\t A\theta)\in Sp_\p(\bbC^{2d})$.

        It is easy to check that $\bar R^{-1}=c(\bar{A}\theta)$. Moreover, by (\ref{produ1}),
	\begin{equation}
		\bar R^{-1}R = c(\bar{A}\theta)c(A\theta) = c(\bar{A}\t A\theta).
	\end{equation}
        Therefore, $\bar R^{-1}R \in Sp_\p(\bbC^{2d})$.
\end{proof}

\begin{proof}[Proof of Proposition \ref{matrix-decomp}.]
  By Prop. \ref{proi2}, $\bar R^{-1}R\in Sp_\p(\bbC^{2d})$. By Prop.
\ref{proi0},  $S:=\sqrt{\bar R^{-1}R}\in Sp_\p(\bbC^{2d})$. Clearly, $\bar R\in Sp(\bbC^{2d})$.  Hence, $T:=\bar RS\in Sp(\bbC^{2d})$.

	\begin{equation}
		\bar{T} = \bar{\bar{R} S}
				= R S^{-1}=R S^{-2}S
			= R\, R^{-1} \bar{R}\, S = T.
	\end{equation}
        Therefore, $T:=\bar RS\in Sp(\bbR^{2d})$.
\end{proof}

\bet   The map
      \[\Osc_\+^\nor(\bbC^{2d})\ni
      \pm\sqrt{\det(\one+A\theta)}\Op(\e^{-A})\mapsto
         c(A\theta)\in Sp_\+ (\bbC^{2d})\]
        is a 2-1 epimorphism of semigroups.
\eet
\proof
We use Proposition \ref{semig} and Theorem \ref{bijections-th}~(2). \qed

\section{Metaplectic group}

It is easy to see that if $C\in\Sym(\bbR^{2d})$, then $c(C\omega)\in Sp(\bbR^{2d})$. In fact, elements of this form constitute an open dense subset of $Sp(\bbR^{2d})$.

We define $Mp(\bbR^{2d})$, called the metaplectic group in dimension $2d$,
to be the group generated by operators of the form
\beq\pm\sqrt{\det(\one+ C\omega)}\Op(\e^{-\i C}),\quad C\in\Sym(\bbR^{2d}).\label{meta}\eeq

The theory of the metaplectic group is well known, see e.g.
\cite{DG}, Section 10.3.1. We assume that the reader is familiar with its basic elements. Actually, we have already used it in our proof of
Prop. \ref{prio1} (3).

The theory of the metaplectic group
can be summed up by the following theorem:
\begin{theorem}
  The metaplectic group consists of unitary operators. Operators of the form (\ref{meta}) constitute an open and dense subset of $Mp(\bbR^{2d})$.
  The map
  \beq  \pm\sqrt{\det(\one+C\omega)}\Op(\e^{-\i C})\mapsto c(C\omega)\label{epi1}\eeq
  extends by continuity to a $2-1$ epimorphism $Mp(\bbR^{2d})\to 
  Sp(\bbR^{2d})$
\end{theorem}

\begin{remark}
  For completeness, one should mention some other natural semigroups
  closely related to $\Osc_\+(\bbC^{2d})$:
  \begin{enumerate}
  \item $\Osc_+(\bbC^{2d})$ generated by operators of the form $a\Op(\e^{-A})$ with $A\in\Sym_+(\bbC^{2d})$, $a\in\bbC$;
\item $\Osc_+^\nor(\bbC^{2d})$ generated by operators of the form
$\pm\sqrt{\det(\one+A\theta)}\Op(\e^{-A})$ with 
 $A\in\Sym_+(\bbC^{2d})$.\end{enumerate}
\end{remark}

	\section{Polar decomposition}

        For an operator $V$,  its \emph{absolute value} is defined as
        \beq |V|:=\sqrt{V^*V}.\eeq
        The following theorem provides a formula for
        the absolute value of elements of the oscillator semigroup.

        \bet\label{AtA1}
        Let $A\in \Sym_\+^\qnd(\bbC^{2d})$. Then
            \beq\big|\Op(\e^{-A})\big|=
            \frac{
\sqrt[4]{\det\big(\one+(B\theta)^2\big)}}
            {\sqrt[4]{\det(\one+\bar A\theta A\theta)}}
            \Op(\e^{-B}),\label{pros1}
            \eeq
            where
          \beq
          B = c\Big(\sqrt{c(\bar A\theta)c(A\theta)} \Big)\theta.\label{pros2}\eeq 
Besides, the function
\[\Sym_\+^\qnd(\bbC^{2d})\ni A\mapsto \big|\Op(\e^{-A})\big|
\]
        is smooth.\eet

        \proof By Prop. \ref{proi2},
        $          c(\bar A\theta)c(A\theta)
    =\bar{c(A\theta)}^{-1}c(A\theta)    \in Sp_\p(\bbC^{2d})$.
        Hence, by Prop. \ref{proi0}  we can define
              $        \sqrt{  c(\bar A\theta)c(A\theta)}\in Sp_\p(\bbC^{2d})$.
Therefore, $B$ defined in (\ref{pros2})  belongs to $\Sym_\p^\qnd(\bbR^{2d})$ and satisfies
        $\bar A\t A=B\t B$.
       
                 We have 
       \beq\Op(\e^{-B})^2=\frac{1}{\sqrt{\det(\one+(B\theta)^2}}\Op\big(\e^{-B\t B}\big) .\eeq
            Hence,
            \begin{eqnarray}
              \Op(\e^{-A})^*\Op(\e^{-A})
                        &=&\frac{\sqrt{\det\big(\one+(B\theta)^2\big)}}
                   {\sqrt{\det(\one+\bar A\theta A\theta)}}
                   \Op(\e^{-B})^2.
            \end{eqnarray}
Besides, $\Op(\e^{-B})\geq0$.            Therefore,
            $\big|\Op(\e^{-A})\big|$ is given by (\ref{pros1}).

Now  the square root is a smooth function on the set of invertible matrices (and obviously on the set of nonzero numbers).
In the formula
(\ref{pros2}) for $A\in \Sym_\+^\qnd(\bbC^{2d})$ we never need to take roots of zero or of non-invertible matrices, because $\one\pm A\theta$ and $\one\pm\bar A\theta$ are invertible.
Therefore,
\beq  \Sym_\+^\qnd(\bbC^{2d})\ni A\mapsto
\sqrt{c(\bar A\theta)c(A\theta)}\eeq
is smooth.
 Therefore, the map $A\mapsto B$ is smooth

 For  $A\in \Sym_\+^\qnd(\bbC^{2d})$, $\bar A,B\in \Sym_\+^\qnd(\bbC^{2d})$. Therefore, by Prop. \ref{prio},
$\one+\bar A\theta A\theta$ and $\one+(B\theta)^2$ are invertible. Hence,
 the prefactors of (\ref{pros1}) are smooth. This ends the proof of the smoothness of
 (\ref{pros1}).
\qed

Let $V$ be a closed operator such that ${\rm Ker} V={\rm Ker} V^*=\{0\}$. Then it is well known that there exists a unique unitary operator $U$
such that we have the identity
\beq V=U|V|.\eeq
called the \emph{polar decomposition}.

\begin{theorem}
	  Let $A\in \Sym_\+^\qnd(\bbC^{2d})$. Let $B\in \Sym_\p^\qnd(\bbC^{2d})$ be defined as in
          (\ref{pros2}). Then        
          \beq
          \Big|\sqrt{\det(\one+A\theta)}\Op(\e^{-A})\Big|=
          \sqrt{\det(\one+B\theta)}\Op(\e^{-B}),\label{meta1}\eeq
          and the unitary operator $U$ that appears in the polar decomposition
          \beq
          \sqrt{\det(\one+A\theta)}\Op(\e^{-A})
          =
          U\sqrt{\det(\one+B\theta)}\Op(\e^{-B})\eeq
          belongs to $Mp(\bbR^{2d})$. Besides, if
          \beq \i C:=A\t(-B)\eeq
          is well defined, then
          \beq
          U=\epsilon\sqrt{\det(\one+C\omega)}\Op(\e^{-\i C}),\eeq
          where $\epsilon=1$ or $\epsilon=-1$. 
\end{theorem}

\begin{proof} By (\ref{prio}),
  \begin{align}
    \one+\bar A\theta A\theta&=(\one+\bar A\theta)(\one+\bar A\t A\theta)^{-1}(\one +A\theta), \\
\one+B\theta B\theta&=(\one+B\theta)(\one+B\t B\theta)^{-1}(\one +B\theta),     
\end{align} Besides, $\bar A\t A= B\t B$. This together with
  (\ref{pros1}) implies  
(\ref{meta1}).

  Assume now that
$ \i C:=A\t(-B)$ is well defined. 
Then clearly	\begin{align}&
\sqrt{\det(\one+A\theta)}\Op(\e^{-A})\\ = &\epsilon
\sqrt{\det(\one+B\theta)}\Op(\e^{-B})\sqrt{\det(\one+\i C\theta)}\Op(\e^{-\i C}).
\end{align}

It remains to show that $\i C$ is purely imaginary.
\begin{align}
  \bar{A\t (-B)}\quad=\quad(-B)\t \bar A&=
  (-B)\t \bar A\t A(-A)\\
  &=(-B)\t B\t B\t(-A)\\
  &=B\t(-A)\quad=\quad-A\t(-B).\end{align}
\end{proof}

\section{Trace and the trace norm}
            
Suppose  we have an operator $K$ on $\mathrm{L}^2(\bbR^n)$.
 As proven in \cite{D} (for a more general setting, see \cite{B1} and \cite{B2}), if $K$ has a continuous kernel $K(x,\,y)$ belonging to $\mathrm{L}^2(\bbR^n\times \bbR^n)$ and $x\mapsto Kx,\ x)$ is in $\mathrm{L}^1(\bbR^n)$, then 
	\begin{equation}
		\Tr K = \int K(x,\,x)\,\mathrm{d}x.
	\end{equation}

        In the case of Weyl--Wigner quantization, for a symbol $a$ we get
	\begin{equation}
		\Tr\Op(a) = \int \Op(a)(x,\,x)\,\mathrm{d}x = (2\pi)^{-d}\int a(x,\,\xi)\,\mathrm{d}x\,\mathrm{d}\xi.
	\end{equation}
	This easily implies the followig proposition:
	\begin{proposition}\label{trace-formula}
		The trace of operator $\Op(\e^{- A})$ with $A\in\Sym_\+(\bbC^{2d})$ is
		\begin{equation}
		  \Tr\Op(\e^{-A}) = \frac{1}{2^{d}\sqrt{\det A}}
                  =\frac{1}{2^{d}\sqrt{\det A\theta}}.
		\label{deter}\end{equation}
	\end{proposition}
        (Note that $\det\theta=1$, hence we could insert $\theta$ in (\ref{deter})).
        
	One can also compute the trace of the absolute value of elements of the oscillator semigroup, the so-called trace norm.
	
	\begin{theorem}
		The trace norm of $\Op(\e^{- A})$, where $A\in\Sym_\+(\bbC^{2d})$, is 
		\begin{equation}\label{norma-eq}
		\Tr|\Op(\e^{- A})| = \frac{\sqrt{2}}{2^d\sqrt{\det|(\one + A\theta)(\one - \sqrt{c(A^*\theta)c(A\theta)})|}}.
		\end{equation}
        \end{theorem}

        \begin{proof} (\ref{deter}) and (\ref{pros1}) imply
           \beq\Tr\big|\Op(\e^{-A})\big|=
            \frac{
\sqrt[4]{\det\big(\one+(B\theta)^2\big)}}
            {2^d\sqrt[4]{\det(\one+\bar A\theta A\theta)(B\theta)^2}}
            .\eeq

            Now, easy algebra shows that
            \begin{eqnarray}
              &&
             \frac{ \det\big(\one+(B\theta)^2\big)}
               {\det(\one+\bar A\theta A\theta)(B\theta)^2}\\
               &=&
               \frac{2 \det\big(\one+c(\bar A\theta)c(A\theta)\big)}
                    {\det(\one+\bar A\theta A\theta)\big(\one
        -\sqrt{c(\bar A\theta)c(A\theta)}\big)^2}\\
&=&\frac{4}{\det(\one+\bar A\theta)(\one+A\theta)
\big(\one
-\sqrt{c(\bar A\theta)c(A\theta)}\big)^2}\\&            =&\frac{2^2}{\Big(\det\Big|(\one+A\theta)
                      \Big(\one
-\sqrt{c(\bar A\theta)c(A\theta)}\Big)\Big|\Big)^2}.
                 \end{eqnarray}
            	\end{proof}

        	\begin{corollary}\label{tracenormreal}
		The trace norm of $\Op(\e^{- B})$, where $B\in\Sym_\+(\bbR^{2d})$, is 
	
\beq
\Tr|\Op(\e^{- B})| = \frac{\sqrt{2}}{2^d\sqrt{\det
    \big||\one+B\theta|-|\one-B\theta|\big|}}.\eeq
        \end{corollary}	

                Thus, if we diagonalize simultaneously
                $B$ and $\omega$, as in the proof of
Prop. \ref{prio1}, then
\beq
\Tr|\Op(\e^{- B})| = \frac{\sqrt2}{4^d\mathop{\prod}\limits_{\lambda_i<1}\lambda_i}.\eeq

\section{Operator norm}

\begin{proposition}
  Let $B\in \Sym_\+(\bbR^{2d})$. Then
  \beq\big\|\Op(\e^{-B})\big\|=\frac{1}{\sqrt{\det(\one+\sqrt{B\theta B\theta})}}.\label{prodi1}\eeq
\end{proposition}
\begin{proof} First, using
(\ref{propo}),
we check that  in the case of one degree of freedom we have
\beq\big\|\Op\big(\e^{-\lambda(x^2+p^2)}\big)\big\|=
\frac{1}{1+\lambda}.\eeq
An arbitrary $B$ we can diagonalize together with $\theta$, as in
the proof of Prop. \ref{prio1} (3),
and then we obtain
\beq
\big\|\Op(\e^{-B})\big\|=\prod_{i=1}^d\frac{1}{1+\lambda_i}.\label{prodi}\eeq
Now the rhs of (\ref{prodi}) can be rewritten as the rhs of
(\ref{prodi1}). \end{proof}

Using \ref{meta1}, we obtain an identity for an arbitrary element of the oscillator semigroup. A closely related result is described in Thm 5.2 of \cite{V}.

\begin{theorem}
  Let $A\in \Sym_\+(\bbC^{2d})$. Then
  \begin{align}
    \big\|\sqrt{\det(\one+A\theta)}\Op(\e^{-A})\big\|&\notag\\
   &\hspace{-15ex} =\frac{
    \sqrt{\det\Big(\one+c\big(\sqrt{c(\bar A\theta)c(A\theta)}\big)\Big)}}
               {\sqrt{\det\Big(\one+\sqrt{c
                   \Big(\sqrt{c(\bar A\theta)c(A\theta)}\Big)
                   c\Big(\sqrt{c(\bar A\theta)c(A\theta)}\Big)}\Big)}}
               .\label{prodi2}\end{align}
\end{theorem}

        \section{One degree of freedom}

In the case of one degree of freedom we have a complete characterization of quantum nondegenerate symmetric matrices.
        
        \bet \label{qnd2} Let $A\in\Sym(\bbC^2)$. Then
        $A\in\Sym^\qnd(\bbC^2)$ iff $\det A\ne1$.
        \eet
        \proof
        We easily compute that for
        $A\in\Sym(\bbC^2)$,
        \beq
        \det(\one+A\theta)=1-\det A.\eeq
        \qed

        Next we describe the quantum degenerate case for one degree of freedom
        on the level of the oscillator group.

        \bet
        Elements of $\Osc_\+(\bbC^2)$ that
are not proportional to an element of $\Osc_\+^\nor(\bbC^2)$ 
are proportional to a projection. They
have the integral kernel of the form
		\begin{equation}
		 c\e^{-(ax^2+by^2)},
\label{pro3}		\end{equation}
where $a,b,c\in\bbC$, $\Re a,\Re b>0$. The Weyl symbol of the operator with the kernel (\ref{pro3}) is
        \beq
	  c\frac{2\sqrt{\pi}}{\sqrt{a+b}}\,\e^{- A},\eeq
          where
          \beq
          A=\frac{1}{(a+b)}\begin{bmatrix}
          4ab&\i(-a+b)\\
          \i(-a+b)&1\end{bmatrix}.\label{pron4}\eeq
          Matrices of the form (\ref{pron4}) with
           $\Re a,\Re b>0$ are precisely all matrices in
          \beq
          \Sym_\+(\bbC^2)\backslash
        \Sym_\+^\qnd(\bbC^2).\eeq

\eet

\section{Application to the boundedness of pseudo-differential operators}




Cordes proved the following result  \cite{C}:

\begin{theorem} 
  Suppose $a\in\mathscr{S}'(\bbR^d\times \bbR^d)$ and $s>\frac{d}{2}$.  
  Then there exists a constant $c_{d,s}$ such that
  \begin{equation}
    \| \Op(a)\|\leqslant c_{d,s}
    \|(1-\Delta_x)^s(1-\Delta_p)^sa\|_\infty.
\end{equation}
\label{cvai}\end{theorem}

The above result can be called the \emph{Calder\'{o}n and Vaillancourt Theorem for the Weyl quantization}. (The original result of Calder\'{o}n and Vaillancourt \cite{CV} concerned the $x-p$ quantization, known also as the standard or Kohn-Nirenberg quantization). 

Note that Theorem \ref{cvai} is not optimal with respect to the number of derivatives. 	The optimal bound on the number of derivatives for the Weyl quantization is $s>\frac{d}{4}$. It was discovered by A. Boulkhemair \cite{Bo} and it requires a different proof than the one developed by Cordes.

In what follows we will describe  a proof of Theorem \ref{cvai}
which gives an estimate of  $c_{d,s}$.
We  will follow the ideas of Cordes and Kato
(\cite{C} and \cite{K}), who however do not give  an explicit bound on the constant  $c_{d,s}$.
The estimate (\ref{propo2}) for the trace norm of operators with Gaussian symbols
 plays an important role in our proof.

We start with the following proposition.
\begin{proposition}
For $s>\frac{d}{2}$, define the functions
\begin{align}
\psi_s(\xi)&: = (2\pi)^{-d}\int\mathrm{d}\zeta\, (1+\zeta^2)^{-s}\e^{\i\zeta \xi},\\
P_s(x,p)&:=\psi_s(x)\psi_s(p).
\end{align}
Then $\Op(P_s)$ is trace class and
\beq\Tr\Big|\Op(P_s)\Big|\leq\frac{\Gamma(s)^2+\Gamma(s-\frac{d}{2})^2}{(2\pi)^d
  \Gamma(s)^2}.\label{esti}\eeq
\label{esti1}
\end{proposition}

\proof
Let us use the so-called Schwinger parametrization
	\begin{equation}
	X^{-s} = \frac{1}{\Gamma(s)}\int_0^\infty \e^{-tX}t^{s-1}\,\mathrm{d}t
	\end{equation}
	to get
	\begin{eqnarray}
	\psi_s(\xi) &=&  \frac{1}{\Gamma(s)(2\pi)^{d}}\int_0^\infty\mathrm{d}t\int\mathrm{d}\zeta\, \e^{-t(1+\zeta^2)}t^{s-1}\e^{\i\zeta \xi}\nonumber\\
	&=& \frac{1}{\pi^{\frac{d}{2}}2^d\Gamma(s)}\int_0^\infty \mathrm{d}t\,t^{s-\frac{d}{2}-1}
        \e^{-t-\frac{\xi^2}{4t}}.
	\end{eqnarray}

        Now
        \beq
        P_s(x,p)
        =\frac{1}{\pi^d2^{2d}\Gamma^2(s)}\int_0^\infty\d u
        \int_0^\infty\d v\e^{-u-v-\frac{x^2}{4u}-\frac{p^2}{4v}}(uv)^{s-\frac{d}{2}-1}
        .\eeq
	By  (\ref{propo2}), we have
        \beq\Tr\left|\Op\big(\e^{-\alpha x^2-\beta p^2)}\big)\right|=\begin{cases}
\frac{1}{(2\sqrt{\alpha\beta})^{d}},&\alpha\beta \leq1,\\
\frac{1}{2^d},&1\leq \alpha\beta 
.\end{cases}\label{propo2a}\eeq
        Hence,
        \begin{align}
          &\Tr\Big|\Op(P_s)\Big|          \\
          \leq&
\frac{1}{2^{2d}\pi^d\Gamma^2(s)}\int_0^\infty\d u
\int_0^\infty\d v\e^{-u-v}\Tr\Big|\Op\Big(
\e^{-  \frac{x^2}{4u}-\frac{p^2}{4v}}\Big)\Big|(uv)^{s-\frac{d}{2}-1}
\\
\leq&
\frac{1}{2^{d}\pi^d\Gamma^2(s)}
\Bigg(\underset{4\leq uv,\ u,v>0}{\int\hspace{-1ex}\d u\hspace{-1ex}\int\hspace{-1ex}\d v}
\e^{-u-v}
(uv)^{s-1}
+
\underset{uv\leq 4,\ u,v>0}{\int\hspace{-1ex}\d u\hspace{-1ex}\int\hspace{-1ex}\d v}
\e^{-u-v}
(uv)^{s-\frac{d}{2}-1}\Bigg)\label{esti2}\\
\leq&\frac{\Gamma(s)^2+\Gamma(s-\frac{d}{2})^2}{2^{d}\pi^d\Gamma^2(s)}.
\end{align}
\qed

\begin{proposition}
  Let $B$ be a self-adjoint trace class operator and $h\in L^\infty(\bbR^{2d})$. Then
  \beq
  C:=\frac{1}{(2\pi)^d}\int\d y\int\d w h(y,w)\e^{-\i y\hat p+\i w\hat x}
  B\e^{\i y\hat p-\i w\hat x}\eeq
  is bounded and
  \beq
  \|C\|
  \leq \Tr|B|\, \|h\|_\infty.\eeq\label{esti9}\end{proposition}
\proof
For $\Phi\in L^2(\bbR^d)$, $\|\Phi\|=1$, define
$T_\Phi:L^2(\bbR^d)\to L^2(\bbR^{2d})$ by
\beq T_\Phi\Theta(y,w):=(2\pi)^{-\frac{d}{2}}(\Phi|
\e^{\i y\hat p-\i w\hat x}\Theta),\quad\Theta\in L^2(\bbR^{2d}).\eeq
We check that $T_\Phi$ is an isometry. This implies that for $\Phi,\Psi\in L^2(\bbR^d)$ of norm one
  \beq
  \frac{1}{(2\pi)^d}\int\d y\int\d w h(y,w)\e^{-\i y\hat p+\i w\hat x}
  |\Phi)(\Psi|\e^{\i y\hat p-\i w\hat x}\label{proi}\eeq
  is bounded and its norm is less than $\|h\|_\infty$. Indeed,
  (\ref{proi}) can be written as the product of three operators
  \beq T_{\Phi}^* hT_{\Psi},\eeq
  where $h$ is meant to be the operator of the multiplication  by the function
  $h$ on the space $L^2(\bbR^{2d})$. Now it suffices to write
  \beq B=\sum_{i=1}^\infty \lambda_i|\Phi_i)(\Psi_i|,\eeq
  where $\Phi_i,\Psi_i$ are normalized, $\lambda_i\geq0$ and $\Tr| B|=\sum\limits_{i=1}^\infty\lambda_i$.  \qed
  
\noindent{\it Proof of Theorem \ref{cvai}.}  Set
\beq
h:=(1-\Delta_x)^s(1-\Delta_p)^sa.\eeq
Then
\begin{align}
  a(x,p)&=(1-\Delta_x)^{-s}(1-\Delta_p)^{-s}h(x,p)\\
  &=\int\d y\int\d w\, P_s(x-y,p-w)h(y,w).
\end{align}
Hence
\begin{align}
  \Op(a)&=\int\d y\int\d w\, \Op\big(P_s(x-y,p-w)\big)h(y,w)\\
  &=\frac{1}{(2\pi)^d}\int\d y\int\d w\, h(y,w)\e^{-\i y\hat p+\i w\hat x}
  \Op(P_s)\e^{\i y\hat p-\i w\hat x}.
\end{align}
Therefore,
by Proposition \ref{esti9},
\beq
\|\Op(a)\|\leq\Tr\big|\Op(P_s)\big|\|h\|_\infty.\eeq
Thus we can set
\beq c_{d,s}=\Tr\big|\Op(P_s)\big|,\eeq
which is finite by Proposition \ref{esti1}.
\qed

Proposition \ref{esti1} yields an explicit estimate for $c_{d,s}$ given by
the rhs of (\ref{esti}). Actually, in the proof of
Proposition \ref{esti1} we have an even  better, although more complicated explicit estimate given by
(\ref{esti2}).

\section{Complex symplectic Lie algebra}

The well known \emph{symplectic Lie algebra} in dimension $2d$ is defined
as the set of $R\in L(\bbR^{2d})$ satisfying
          \beq R^\t\omega+ \omega R=0.\label{insympl}\eeq
Similarly, the set of
$R\in L(\bbC^{2d})$ satisfying (\ref{insympl}) is called the
\emph{complex symplectic Lie algebra} in dimension $2d$ and denoted
 $sp(\bbC^{2d})$. As usual in the complex case, we  usually  prefer to replace
          $\omega$ in (\ref{insympl}) with $\theta$.

We define
\begin{align}
  sp_+(\bbC^{2d})&:=\{D\in sp(\bbC^{2d})\mid
  D^*\theta+\theta D\geq0\},\\
    sp_\+(\bbC^{2d})&:=\{D\in sp(\bbC^{2d})\mid
    D^*\theta+\theta D>0\}.
    \end{align}  
We also introduce
                    \begin{eqnarray}
sp_\h(\bbC^{2d})&:=&\{D\in
sp(\bbC^{2d}) \mid \bar D=-D\},\\
                       sp_{\mathrm{p}}(\bbC^{2d})&:=&\{D\in
 sp_{\h}(\bbC^{2d}) \mid \theta D >0\}.\end{eqnarray}

\bep \ben\item Let $D\in sp(\bbC^{2d})$. Then $\e^{-D}\in  Sp(\bbC^{2d})$. 
\item Let $D\in sp_\+(\bbC^{2d})$. Then $\e^{-D}\in  Sp_\+(\bbC^{2d})$.
\item Let
  $D\in sp_\h(\bbC^{2d})$. Then $\e^{-D}\in Sp_\h(\bbC^{2d})$.
\item             Let $D\in sp_{\mathrm{p}}(\bbC^{2d})$. Then $\e^{-D}\in
Sp_{\mathrm{p}}(\bbC^{2d})$.
\een\label{pru}\eep

\proof 
(1) and (3) are obvious corollaries from the definitions.

(2): Integrating
\beq\frac{\d}{\d t}(\e^{-tD})^*\theta\e^{-tD}=-(\e^{-tD})^*(D^*\theta+\theta D)\e^{-tD}<0\eeq
    we obtain $(\e^{-D})^*\theta\e^{-D}<\theta.$

(4): We can write
 \[\e^{-D}=\e^{-\theta(\theta D)}.\]
 We diagonalize simultaneously the positive form $\theta D$ and $\theta$.
 In the diagonalizing basis, the matrices $\theta$ and $\theta D$ commute, the former has eigenvalues $\pm1$, the latter has positive eigenvalues.  Hence $\e^{-D}$ has positive eigenvalues.
 \qed

 \section{Hamiltonians}

 Let $H\in \Sym(\bbC^{2d})$. As usual, the quadratic form  $\bbR^{2d}\ni y\mapsto y^\t Hy\in \bbC$ will be also denoted by $H$.
Let us briefly recall the properties of quantum quadratic Hamiltonians $\Op(H)$ and their relationship to the metaplectic group. We will use \cite{DG} as the basic reference, although most of these facts are well known.

 Set \beq
 D:=2H\omega^{-1}. \label{gene}
 \eeq
 Clearly, $D\in sp(\bbC^{2d})$. We will say that $D$ is the \emph{symplectic generator associated with the Hamiltonian $H$}.

First assume that
 $H\in\Sym(\bbR^{2d})$. It is well known that then $\Op(H)$ is essentially
 self-adjoint on $\cS(\bbR^d)$ (see e.g. \cite{DG} Thm 10.21).
 Moreover,  $\e^{\i t\Op(H)}\in Mp(\bbR^{2d})$
 (see e.g. \cite{DG} Thm 10.36). Under the epimorphism \ref{epi1}, 
 $\e^{\i t\Op(H)}$ is mapped onto $\e^{tD}$, where $D\in sp(\bbR^{2d})$ is defined by (\ref{gene}) (see e.g. \cite{DG} Thm 10.22).
 Finally, if $\e^{tD}\in Sp^\reg(\bbR)$ and $C_t:=c(\e^{tD})\omega^{-1}$,
 \beq \e^{\i t\Op(H)}=\sqrt{\det(1+C_t\omega)}\Op(\e^{-\i C_t}), \eeq
 see e.g. \cite{DG} Thm 10.35.

 Next consider
 $H\in\Sym_\+(\bbC^{2d})$.
 It is easy to show that $\Op(H)$
 extends from $\cS(\bbR^d)$ to a maximal accretive operator (see e.g. \cite{DG}, Thm 10.21). Moreover,
 $\e^{- t\Op(H)}\in \Osc_\+^\nor(\bbC^{2d})$. In fact, if $D$
 is defined as in (\ref{gene}), then $-\i D\in sp_\+(\bbC^{2d})$, and hence
 by Prop. \ref{pru} (2),
 $\e^{\i tD}\in Sp_\+(\bbC^{2d})$. Moreover, 
 under the epimorphism (\ref{epi2}),
 $\e^{- t\Op(H)}$ is mapped onto $\e^{\i tD}$. Finally,
if we set $A_t:=c(\e^{\i tD})\theta$, then
 \beq
 \e^{-t\Op(H)}=\sqrt{\det(\one+A_t\theta)}\Op(\e^{-A_t}),\eeq
see e.g. in \cite{DG} Thm 10.35.

   \section{Holomorphic 1-parameter subsemigroups}

   Let $H\in \Sym_\+(\bbC^{2d})$. As we recalled above, $\Op(H)$ is maximally accretive, and hence
   \beq[0,\infty[\ni t\mapsto \e^{-t\Op(H)}\eeq
         is a well defined subsemigroup of $\Osc_\+(\bbC^{2d})$. One can ask whether it can be extended to a larger subsemigroup if we replace real $t$ with a complex parameter.

         If  $H$ is real, then the answer is obvious and simple. Then $\Op(H)$ is a positive self-adjoint operator and we have a well defined semigroup
         \beq\{z\in \bbC\mid \Re z\geq0\}\ni z\mapsto \e^{-z\Op(H)}\label{prau}\eeq
         inside $\Osc_+(\bbC^{2d})$.
For  $ \Re z>0$, (\ref{prau}) is in
$\Osc_\+(\bbC^{2d})$.

If $H$ is not real, then the answer can be more complicated.

Let $D\in sp_\+(\bbC^{2d})$
correspond to $H$ as in (\ref{gene}). Clearly
   \beq
\bbC\ni   z\mapsto\e^{\i zD}\in Sp(\bbC^{2d})\eeq
is a holomorphic subgroup of $Sp(\bbC^{2d})$. However,
not all  elements of the complex symplectic group correspond to (bounded) operators on the Hilbert space. 
        Motivated by this, we define
        \begin{align}
          \cA_+(H)&:=\{z\in\bbC\mid\e^{\i zD}\in Sp_+(\bbC^{2d})\},\\
                    \cA_\+(H)&:=\{z\in\bbC\mid\e^{\i zD}\in Sp_\+(\bbC^{2d})\}.\end{align}
        From the definition it is obvious that
        $\cA_+(H)$ is a closed subsemigroup of $\bbC$ and
                $\cA_\+(H)$ is an open subsemigroup of $\cA_+(H)$.

        If $z\in \cA_\+(H)$, then we  define
        \begin{align}
          A_z&:=c(\e^{\i z D})\theta\in\Sym_\+(\bbC^{2d}),\\
          \e^{-z\Op(H)}&:=\sqrt{\det(\one+ A_z\theta)}\Op\big(\e^{-A_z}\big).\label{expon}\end{align}
        (The definition of (\ref{expon}) is consistent with the usual definition of $\e^{-z\Op(H)}$ for real positive $z$).

         The shapes of $\cA_+(H)$ and $\cA_\+(H)$ can be
         quite curious. This is already seen in the simplest nontrivial example, known under the name of the \emph{Davies harmonic oscillator}, as shown in \cite{AV}, see also \cite{V}. In this example, $\psi\in]-\frac{\pi}{2},\frac{\pi}{2}[$ is a parameter, the classical and quantum Hamiltonians and the generator are
       \begin{align}
         H_\psi&:=\e^{\i\psi} x^2+\e^{-\i\psi} p^2,\\
   \hat H_\psi:=        \Op(H_\psi)&=\e^{\i\psi} \hat x^2+\e^{-\i\psi} \hat p^2,\\
                  D_\psi&:=2\begin{bmatrix}0&-\e^{\i\psi}\\\e^{-\i\psi}&0\end{bmatrix}.
       \end{align}
The proposition below reproduces the result of Aleman and Viola (see (1.2) of \cite{AV}).
       \begin{proposition}
       	Let $H_\psi$ be the Davies' harmonic oscillator, as above. Then:
       	\begin{equation}
       	\cA_+(H_\psi) = \{z\in\bbC\ |\ \Re(z)\geq 0 \mbox{ and } |\arg \tanh z| + |\psi| \leqslant\frac{\pi}{2}\},
       	\end{equation}
       	\begin{equation}
       	\cA_\+(H_\psi) = \{z\in\bbC\ |\ \Re(z)>0 \mbox{ and } |\arg \tanh z| + |\psi| <\frac{\pi}{2}\}.
       	\end{equation}
       \end{proposition}
   \begin{proof} 
       $\i D_\psi$ generates a holomorphic group in $Sp(\bbC^{2d})$, which can be computed using
        $D_\psi^2=-4\one$:
       \beq
       \e^{\i z D_\psi}=\begin{bmatrix}
       \cosh 2z&\i\e^{\i\psi}\sinh2 z\\-\i\e^{-\i\psi}\sinh2 z&\cosh2 z\end{bmatrix}
       \eeq
       Now
       \beq A_{\psi,z}=c(\e^{\i z D_\psi})\theta=2 \tanh z\begin{bmatrix}\e^{-\i\psi}&0 \\ 0 & \e^{\i\psi} \end{bmatrix}.
       \eeq

	Let us denote $t:= \arg \tanh z$. $A_{\psi,z}$ belongs to $\Sym_\+(\bbC^{2d})$    iff $\Re(z)>0$ and
	\begin{equation}
	\begin{cases}
	|t+\psi|<\frac{\pi}{2}, \\ |t-\psi|<\frac{\pi}{2}.
	\end{cases}
	\end{equation}
	The above pair of inequalities is equivalent to
	\begin{equation}
	|t|+|\psi|<\frac{\pi}{2}.
	\end{equation}
	
	By Theorem \ref{bijections-th}(2),
 $A_{\psi,z}\in\Sym_\+(\bbC^{2d})$   
        iff  $\e^{\i zD_\psi}\in Sp_\+(\bbC^{2d})$.

        The proof for $\cA_+(H_\psi)$ is analogous.
    \end{proof}

\end{document}